\documentclass[a4paper,onecolumn,10pt]{llncs}

\pdfoutput=1

\usepackage{url}
\usepackage{a4wide}
\usepackage{xfrac}
\usepackage[final]{pdfpages}
\usepackage{amsmath}
\usepackage{amsfonts}
\usepackage{wasysym}

\title{BSP vs MapReduce\thanks{Appeared at ICCS 2012}}
\author{Matthew Felice Pace\thanks{Research supported by the Centre for Discrete Mathematics and its Applications (DIMAP), University of Warwick, EPSRC award EP/D063191/1}}
\institute{DIMAP and Department of Computer Science\\
University of Warwick\\
Coventry, CV4 7AL, UK\\
\email{matthewfp@dcs.warwick.ac.uk}}

\begin{document}
\maketitle

\begin{abstract}
The MapReduce framework has been generating a lot of interest in a wide range of areas. It has been widely adopted in industry and has been used to solve a number of non-trivial problems in academia. Putting MapReduce on strong theoretical foundations is crucial in understanding its capabilities. This work links MapReduce to the BSP model of computation, underlining the relevance of BSP to modern parallel algorithm design and defining a subclass of BSP algorithms that can be efficiently implemented in MapReduce.
\end{abstract}

\section{Introduction}
Efficient algorithms are of fundamental importance in a world driven by computation. The amount of data in today's world is constantly increasing at a dramatic rate, and this poses a great challenge to algorithm design. The web graph today is much larger than could have probably been imagined 20 or even 10 years ago. Other examples of such large datasets abound, from search logs to financial data.

\paragraph{}
One framework of processing massive datasets, called MapReduce, has been attracting great interest. It was developed and widely used by Google \cite{MapReduceOrig}, while its open source implementation Hadoop \cite{Hadoop} is currently used by more than 100 companies worldwide including eBay, IBM, Yahoo!, Facebook, and Twitter, along with a number of universities \cite{HadoopUsers}. Some companies, such as Amazon and Microsoft, allow users to run MapReduce programs on their cloud \cite{AmazonElastic,Daytona}. MapReduce algorithms have been used to solve a number of non-trivial problems in diverse areas including data processing, data mining and graph analysis (see \cite{MapReduceAlg1,MapReduceAlg2} for an introductory list of works).

\paragraph{}
The popularity of MapReduce means that it now plays a prominent role in the field of parallel computation. This creates the demand to put MapReduce on sound theoretical foundations, and to establish its relationship to other major parallel computation models such as BSP and PRAM. In this regard, major work has already been done by Feldman et al. \cite{Feldman}, Karloff et al. \cite{Karloff} and Goodrich et al. \cite{Goodrich}. Between them, these authors have shown that MapReduce can be used to simulate the PRAM and BSP model, as well as link MapReduce to a subclass of streaming algorithms. 

\paragraph{}
This work aims to further analyse the connection between MapReduce and BSP. Section 2 introduces the MapReduce framework, followed by an introduction to the BSP model in section 3. In section 4, a model for MapReduce is developed and the links between this model and BSP are discussed. A simulation of MapReduce on BSP is given in section 5, while section 6 defines a subclass of algorithms that can be efficiently implemented in MapReduce along with a few examples. The final section provides some concluding views.

\section{MapReduce framework}
Parallel programming is, unfortunately, not straightforward. The developer needs to cater for a number of non-trivial problems such as fault tolerance, load balancing and synchronisation. In their seminal 2004 paper, Dean and Ghemawat \cite{MapReduceOrig} introduced the parallel computation framework MapReduce. This framework allows for simplified programming on large clusters of low-end systems. The simplicity of MapReduce is due to the fact that it acts as an abstraction on top of the complex details that need to be catered for when writing parallel code. This allows the programmer to focus on the functions that actually manipulate the data at hand.

\paragraph{}
The main idea for MapReduce comes from functional programming languages such as Lisp \cite{Lisp}. The Lisp function Map takes as input a function and a sequence, and applies this function to all the elements in the sequence. The function Reduce, given a sequence of values and a binary function, uses this function to combine the elements of the sequence into a single output value. These two functions, performed in rounds, form the basis of MapReduce.

\paragraph{}
MapReduce was originally designed to run on large clusters of low-end commodity machines. Every machine has a processor, a fast primary memory and a slower secondary memory, and is connected via an underlying network to the rest of the cluster. The secondary memory is used as part of a global shared memory, with each machine allowed to access other machines' secondary memory remotely, although only during synchronisation. The framework works on data in the form of $\langle key;value \rangle$ pairs, with $n$ initial such pairs stored in global memory as the algorithm's input. A MapReduce algorithm proceeds in rounds, with two phases in every round, a map and a reduce phase. Each phase is composed of an input, a computation and an output step, with the output of each phase used as input to the next phase. When a phase is finished, i.e. when every machine has written its output data to shared memory, the data is synchronised, i.e. each machine is allowed to read the data written in the previous phase. No other communication is allowed between machines, except with the master processor as described later. The local primary memory is cleared before each synchronisation. Parallelism is achieved by having different machines perform the same functions on different data.

\paragraph{}
System failures are common in clusters of hundreds or thousands of low-end systems, and therefore automatic fault tolerance and load balancing play a crucial role in the design of the MapReduce framework. These are achieved by having each machine work on multiple tasks, making it easier to reprocess and reassign these tasks in case of machine failure. One system processor is assigned the role of master processor and controls how these tasks are assigned across the other worker processors. Each task reads its input and processes it using either a map or reduce function designed by the developer. Each function acts on a single $\langle key;value \rangle$ pair, with tasks computing a function for every input pair assigned to them. MapReduce was initially designed to cater for algorithms where the output was much smaller than the input, so the number of map tasks $q$ was much larger than the number of reduce tasks $r$. Typically, in a system with $p$ processors, $q$ was between 10 and 100 times $p$ and $r$ was around 2 to 5 times $p$, as described in the original paper \cite{MapReduceOrig} and the official Hadoop tutorial \cite{HadoopTutorial}. However, subsequently algorithms have been designed that use the same number of map and reduce tasks.

\paragraph{}
The first round of MapReduce proceeds as follows. The algorithm's input is placed in global memory, split into $q$ parts. Each part will be processed by a separate map task, one pair at a time. The master then assigns map tasks to workers. The number of map tasks is typically greater than the number of processors, so one task is initially assigned to each worker processor. When a processor finishes computing its task a new map task is assigned to it. Upon being assigned a new task, a processor reads the input data pertaining to the task from global memory to primary memory and processes it. The output of each map task, in the form of $\langle key;value \rangle$ pairs, is stored in the worker's secondary memory.

\paragraph{}
While in the map phase, every map task processes various $\langle key;value \rangle$ pairs, in the reduce phase all the values for a given key are processed by a single reduce task. This is achieved by logically partitioning the secondary memory of each worker processing a map task into $r$ partitions, and then determining in which particular partition an output pair should be stored, a process accomplished by the shuffle step. This step can be viewed as a data routing step, determining which reduce task will process a data pair based on the pair's key. This function is performed by the workers while processing the map tasks. Typically, a function such as $(hash(key)\bmod r)$ is used, where the $hash$ function is a simple function, computable in a small constant time, used to map the keys to a more manageable domain. Other partitioning functions can be defined by the user, especially if the keys are in numeric form, such as partitioning the keys into $r$ logical partitions representing various ranges of values.

\paragraph{}
When all the map tasks have finished, the $r$ reduce tasks are assigned to the available workers using the same process as for the map tasks. Each reduce task accesses the data assigned to it, stored across the workers responsible for computing the $q$ map tasks. All the pairs with the same key are stored in the same partition, and each partition can have pairs with different keys. All this data is sorted by key and combined such that all values associated with a key are grouped together in a single $\langle key;value \rangle$ pair. This is sometimes considered as being a second part of the shuffle step.

\paragraph{}
Each reduce task then reads its assigned data and processes it one $\langle key;value \rangle$ pair at a time using the reduce function. The task's output is then written to global memory, and can either be the final output of the algorithm or used as input to a new round of MapReduce. The input for a new round of MapReduce is partitioned into $q$ parts by the master processor and the process just described is repeated.

\paragraph{}
The relationship between the system processors and the map and reduce tasks leads to some interesting aspects of the MapReduce framework. A number of map and reduce tasks can be performed, in sequence, by a single processor in each round. In each of the map and reduce phases, tasks are assigned to workers as these finish their previously assigned task. Therefore, if computation is equally divided between tasks, then every processor will perform about $\sfrac{q}{p}$ map tasks and $\sfrac{r}{p}$ reduce tasks. If on the other hand the computation time differs for each task, then load balancing is automatically achieved. It also allows for the efficient handling of fault tolerance. However, the task assignment strategy also places some limitations on the framework. Between rounds the data is split up into $q$ parts and each is assigned to a map function. After any task finishes, the worker's primary memory is cleared, so data cannot be associated with a single processor and accessed at will in different rounds. Therefore, any data that is required in multiple rounds should be specifically stored in global memory.

\paragraph{}
Given a multiset of $n$ $\langle key ; value \rangle$ pairs as an input, the above process describing a single MapReduce round is defined by two functions: map and reduce, and the shuffle step. These are defined as follows:

\begin{itemize}
	\item Given a single input pair from the multiset of the round's input pairs $\{\langle k_1 ; v_1 \rangle, \langle k_2 ; v_2 \rangle, \ldots, \langle k_n ; v_n \rangle\}$, the \textit{map} function performs some computation to produce a new intermediate multiset of $\langle key; value \rangle$ pairs $\{\langle l_1 ; w_1 \rangle, \langle l_2 ; w_2 \rangle, \ldots, \langle l_m ; w_m \rangle\}$.

	\item The union of all the intermediate multisets produced by the \textit{map} functions is acted upon in the \textit{shuffle} step. All the pairs with the same key $l_i$ are combined to produce a new set of lists of the form $\langle l_i;w_1,w_2, \ldots \rangle$.

	\item Each of the lists produced by the shuffle step is passed to a separate \textit{reduce} function that performs some computation to produce a new list $\langle j_i; x_1, x_2, \ldots \rangle$.
\end{itemize}

\section{The BSP Model}
The traditional von Neumann model has served, for a long time, as the main model for designing and reasoning about sequential algorithms. It has also served as a reference model for hardware design. In the context of parallel algorithm design, no such ubiquitous model exists. One of the earliest attempts at defining such a model was the PRAM \cite{JaJa}. While it allowed for better theoretical reasoning about parallel algorithms, PRAM made a number of assumptions that cannot be fulfilled in hardware; mainly because the cost of communication is greater than that of computation and the number of processors is limited.

\paragraph{}
Valiant \cite{Valiant} introduced the BSP model to better reflect the hardware design features of mainstream parallel computers, through the direct mode of BSP (assumed in this paper). The BSP model allows for efficient parallel algorithm design without any overspecification requiring the use of a large number of parameters. The underlying parallel computer implementation is similarly not overspecified. A BSP computer can be defined by $p$ processors, each with its local memory, connected via some means of point-to-point communication. BSP algorithms proceed in supersteps in each of which processors receive input at the beginning, perform some computation asynchronously, and communicate any output at the end. Barrier synchronisation is used at the end of every superstep to synchronise all the $p$ processors in the system. Each processor can communicate directly with every other processor, providing complete control over how the data is distributed between the processors in every superstep.

\paragraph{}
A BSP system is also defined by its bandwidth inefficiency $g$. Every $l$ time steps an attempt is made to synchronise the processors. If they have finished their supersteps then the processors are synchronised, otherwise they keep on working for another $l$ time units.

\paragraph{}
An algorithm designed in BSP can be measured by three main features: the computation time and communication cost for each superstep, and the number of supersteps. Let $w_s$ be the maximum number of arithmetic operations performed by each of the $p$ processors in a superstep $s$. Let $h'_s$ be the maximum input, and $h''_s$ the maximum output data units over all $p$ processors in superstep $s$. The total communication cost of superstep $s$ is $h_s = h'_s + h''_s$. An algorithm running in $S$ supersteps, therefore, has costs $W = \sum_{s=1}^{S}{w_s}$ and $H = \sum_{s=1}^{S}{h_s}$. An estimate running time on any physical system is defined as $W + H \cdot g + S \cdot l$.

\paragraph{}
The synchronisation periodicity $l$ , is typically much higher than the bandwidth inefficiency $g$. Therefore, the number of supersteps required by an algorithm should be minimised as much as possible, while aiming to achieve optimal computation and communication costs. Balanced algorithms are achieved by dividing computation and communication equally amongst the available processors. Since $g$ and $l$ vary between system implementations, algorithm design revolves around improving $W$, $H$ and $S$.
\section{A Model for MapReduce}
The popularity of the MapReduce framework requires that the theoretical limitations of the system be analysed, and theoretical models developed for it. The connection between MapReduce and existent models of parallel algorithm design also needs to be investigated. Pioneering work in the field has already been done, linking MapReduce to other parallel models, as well as to other fields of computer science.

\paragraph{}
The first work \cite{Feldman}, by Feldman et al., proposed a model for the framework aimed at linking it to the data stream model. A subclass of MapReduce algorithms, called mud (massive, unordered and distributed) algorithms, is defined and shown to be closely related to a subclass of streaming algorithms. In fact, mud algorithms can simulate any symmetric (order invariant) streaming algorithms with comparable communication and storage costs. The computation cost is exponential due to the use of Savitch's theorem \cite{Sipser}. Symmetric mud algorithms are limited in space and communication to polylogarithmic cost in terms of the algorithm's input size. The model is very restrictive in the type of algorithms that can be designed, since only one round of MapReduce is allowed and the input is assumed to be independent separate streams.

\paragraph{}
Karloff et al. \cite{Karloff}, propose a model that better captures the specifics of the MapReduce framework. The map and reduce phases are clearly defined, and algorithms can have multiple rounds. A limit is placed on the number of processors in the system, such that the input size $n$ is greater than the number of workers. The primary and secondary memory of each worker is limited to $n^{1-\epsilon}$, for some $\epsilon > 0$, as is the size of the input and output data for each worker in a specific round. The size of global shared memory is in turn bound by the size of the secondary memory of each worker. The authors show that a subclass of EREW PRAM algorithms can be simulated using the model.

\paragraph{}
The most recent model (to the best of our knowledge) is due to Goodrich et al. \cite{Goodrich} which is based on BSP. The size of data that can be sent or received by each reduce task is limited to a value $M$ determined by the algorithm designer. Apart from this restriction, the model differs from BSP in the way the communication cost is calculated. Simulations of CRCW PRAM and BSP algorithms are presented.

\paragraph{}
The idea behind placing limits on I/O size and storage size is to enforce parallelism in the algorithm design, as discussed in \cite{Karloff} and \cite{Goodrich}. If this is not enforced, algorithms could be designed that trivially place all the data on a single processor and run the sequential algorithm. This should not be of any concern if the algorithm designer is aiming to move from sequential to parallel design, and such limitations are not posed in existent models like PRAM and BSP.

\paragraph{}
From the work presented in \cite{Goodrich} and the discussions in the previous two sections it is evident that the relationship between BSP and MapReduce is very strong. Both handle parallel algorithm design in a coarse-grained fashion, interleaving phases of computation and communication. Both can be used to design algorithms running on clusters of low-end systems connected with point-to-point communication, and both make use of synchronisation between rounds/supersteps. The aim of the rest of this work is to further investigate this relationship.

\paragraph{}
Goodrich et al. \cite{Goodrich} showed that all BSP algorithms are implementable in MapReduce by only using the reduce phase, setting the map function to the identity function. However, they did not discuss the efficiency of implementing BSP algorithms on the framework, given its differences to conventional parallel systems. The main difference is that after every map or reduce task finishes, the worker on which they were computed clears the primary local memory. Any data which might need to be used in the following rounds has to be stored in global memory, increasing the communication costs in the process. 

\paragraph{}
Another aspect of MapReduce that differs from BSP is the use of map and reduce tasks running on physical processors. These tasks can be viewed as virtual processors, with multiple virtual processors running on top of physical processors. In the original BSP paper \cite{Valiant}, multiple virtual processors were introduced to design BSP with automatic memory management. Work is divided between the $v$ virtual processors, with each of the $p$ physical processors performing the work of $v/p$ virtual processors. If the work and communication are equally divided between the virtual processors, then work and communication will also be balanced amongst the physical processors. Similarly, if the work is balanced between the map and reduce functions in each MapReduce round, then every processor will process $\sfrac{q}{p}$ map tasks and $\sfrac{r}{p}$ reduce tasks.

\paragraph{}
BSP does not allow any communication between processors in between synchornisation. Similarly, in MapReduce, map and reduce tasks are not allowed to communicate between each other in the same phase. However, asynchronous communication is used between the master processor and the workers to assign these tasks. This allows for dynamic load balancing to be achieved when the individual computation time of the tasks are not balanced. Also, for tasks with unknown exact computation time an offline balanced distribution between processors cannot be found, so dynamic load balancing is used.

\paragraph{}
Finally, while one BSP superstep involves a single computation step, and corresponding input and output communication phases, a round of MapReduce is made up of the map and reduce computation phases with their respective input and output communication. Also, the BSP model assumes that data can be sent directly to any processor in the system, but while data can be directed to specific reduce tasks using specific keys, there is no way of determining which map task will process which part of the input data.

\paragraph{}
Given the similarities between BSP and MapReduce, a model for MapReduce, $MR(p,g,l)$, is proposed, that  is based on BSP. The model has three parameters: the number of physical processors $p$ in the system, the inverse bandwidth $g$ and the system latency $l$ which is the time taken by the shuffle network to set up communication between the map and reduce phases. The number of map tasks $q$ and reduce tasks $r$ to be used are chosen by the algorithm designer, and can vary depending on the algorithm being designed. In specific rounds, a subset of the $r$ reduce tasks can be used instead of all the available tasks. This is not allowed in the case of the map tasks. For the rest of this work, it is assumed that for input size $n$, $n \gg q$ and $n \gg r$, and $q > p$ and $r > p$.

\paragraph{}
Let $q_1, q_2, \ldots$ and $r_1, r_2, \ldots$ denote the specific map and reduce tasks respectively in a single round, and let $R_d$ denote the number of reduce tasks used in a round $d$. The size of data written to or read from shared memory between rounds is denoted by $c_{q_1}, c_{q_2}, \ldots$ for map tasks and $c_{r_1}, c_{r_2}, \ldots$ for reduce tasks. Finally, $t_{q_1}, t_{q_2}, \ldots$ and $t_{r_1}, t_{r_2}, \ldots$ denote the running time of the specific map and reduce tasks respectively.

\paragraph{}
Given the definitions above, for a single round $d$, the computation cost is defined as $T_d = (\sum_{i=1}^{q} t_{q_i}) + (\sum_{i=1}^{R_d} t_{r_i})$. The communication cost for round $d$ is $C_d = (\sum_{i=1}^{q} c_{q_i}) + (\sum_{i=1}^{R_d} c_{r_i})$. A MapReduce algorithm therefore has cost $T + C \cdot g + D \cdot l$, where $T = \sum_{d=1}^{D} T_d$, $C = \sum_{d=1}^{D} C_d$, and $D$ is the number of rounds.

\paragraph{}
It is assumed that since the work done by each task involves at least reading and writing its input and output data from shared memory, then $T \geq C$.

\paragraph{}
The model allows for cost measurements that better reflect the load balancing properties of the MapReduce framework. The MapReduce dynamic load balancing framework leads to a maximum processing time on each processor of $(2 - \sfrac{1}{p}) OPT$, as discussed in \cite{Graham}, where $OPT$ is the optimal maximum processing time on each processor. For balanced algorithms in which the computation and communication costs are equally divided amongst the map and reduce tasks, each machine will handle around $\lceil \sfrac{q}{p} \rceil$ map and $\lceil \sfrac{R_d}{p} \rceil$ reduce tasks in a single round $d$, given that no machines in the cluster fail.

\section{MapReduce on BSP}
The rise in popularity of the MapReduce framework has renewed interest in parallel computing. Still, irrespective of this popularity, it has not been determined whether MapReduce allows for better algorithm design than previously existent models such as BSP, which has been one of the most prominent models for parallel algorithm design for over two decades. It is, therefore, important to ask whether any MapReduce algorithm can be simulated in BSP while preserving its asymptotic costs?

\begin{theorem}
Any round $d$ of an $\mathrm{MR}(p,g,l)$ algorithm with known individual task times and a maximum execution time $t_{max}$ on any processor, can be simulated by a $\mathrm{BSP}(p,g,l)$ machine in $O(1)$ supersteps, such that the maximum execution time on any processor is $O(t_{max})$.
\end{theorem}

\begin{proof}
An MR round can be split into two separate phases, the map phase and the reduce phase. Each phase involves reading input data, performing some computation and outputting any results, in order. These two phases can be simulated by two BSP supersteps $s'$ and $s''$.

\paragraph{}
The $q$ map tasks and $R_d$ reduce tasks in MapReduce are distributed amongst the $p$ processors as these become idle. Since the computation time for each task is known, then a optimal distribution of tasks amongst the processors exists which reduces the maximum processing time on each processor, also known as the makespan. Let this optimal makespan be $t_{OPT}$, then Graham \cite{Graham} states that any distribution of the tasks amongst the processors, as these become idle, will lead to a makespan of $(2 - \sfrac{1}{p}) t_{OPT}$, i.e. $t_{max} = O(t_{OPT})$. Such a distribution can be found offline, by using the same method employed by MapReduce  or using the PTAS presented \cite{Hochbaum}, and the tasks distibuted amongst the BSP processors makespan of $O(t_{max})$.

\paragraph{}
The shuffle phase in MapReduce determines how data is distributed amongst the reduce tasks in the map phase. This can be performed while simulating the $q$ map tasks. The data assigned to each reduce task is then sorted and combined by key in the reduce phase, which can be performed during the simulation of the $R_d$ reduce tasks.

\paragraph{}
If the MR round $d$ is the last round of the algorithm, then the results are stored in global memory and computation stops. When the output of a reduce phase serves as input to a new MR round, the input is stored in global memory and divided into $q$ equal parts by the master processor. Each map task is then assigned one such part. This distribution can be achieved in BSP by adding an extra superstep. Each processor sends the size of its data to a designated master processor, which determines how the data has to be evenly distributed amongst the processors and communicates this to the rest of the processors. The processors then communicate the actual data between them according to the determined distribution. The cost of this procedure is $w = O(p)$, $h = O(p)$ and $s = O(1)$. Since the computation and communication costs of the algorithm are functions of $n$, and $n \gg p$, then the total cost of the algorithm does not increase due to this procedure. \Square
\end{proof}

\paragraph{}
Looking at the simulations and the way the MapReduce framework works it is evident that, for algorithms with known computation time for the individual tasks, designing such algorithms in MapReduce does not provide any asymptotic speed-up over design in BSP. MapReduce forces the designer to cater for features such as the shuffle step, the lack of physical memory storage between rounds, the lack of specifiable association between data and map tasks in the map phase, and the intrinsic distinction between map and reduce tasks and how these process and output data. BSP does not enforce such limitations, but these can be introduced by the designer if necessary, allowing a much more flexible algorithm design.

\paragraph{}
This flexibility provided by BSP over MapReduce makes it a natural choice for algorithm design. Performing algorithm design in BSP also means that the algorithms can be implemented on several different systems, not just MapReduce. It would also be an important step towards consolidating the field of parallel algorithm design to use a single theoretical model. For tasks with unknown execution time, however, MapReduce should still be used since BSP does not have dynamic load balancing capabilities.

\section{BSP on MapReduce}
Having determined that the BSP model can simulate any MapReduce algorithm having tasks with known execution times, and that it allows for better parallel algorithm design, it is important to determine the role of MapReduce. MapReduce offers simplicity in the development process by abstracting intricate programming details. Thus, it is important to determine whether any algorithm design in BSP can be implemented in MapReduce, and if the framework's simplicity comes at a price.

\paragraph{}
Goodrich et al. \cite{Goodrich} give a simulation for any BSP algorithm on their MapReduce model. The memory size $M$ of each processor is limited to $\lceil N/P \rceil$, where $N$ is total memory size of the BSP system and $P$ is the number of BSP processors. The simulation is also correct for the MR model presented in this paper, and works as follows:
\begin{itemize}
	\item The map function is set to the identity function.
	\item Each reduce task simulates a BSP processor, with each reduce phase simulating a single superstep.
	\item Every reduce task receives the input of the BSP processor it is simulating, performs the required computation and outputs the results to the relevant reduce tasks.
	\item Any data that in BSP is stored in local memory between supersteps is stored in global memory and read in the following round.
\end{itemize}

Simulating a BSP algorithm running in $R$ supersteps with $N$ total memory has a cost in MapReduce of $O(R)$ rounds and $O(RN)$ communication. While the result of \cite{Goodrich} is very important, no distinction is made between algorithms whose asymptotic costs are preserved when implemented in MapReduce and those for which the costs do increase. 

\paragraph{}
The previously discussed differences between MapReduce and BSP lead to certain inefficiencies when implementing parallel algorithms designed in BSP on the MapReduce framework. The most evident, as also shown in the above simulation, is that the map phase is mostly ignored, with map tasks assigned only the identity function. This is due to the fact that it is not possible to pass data to specific map tasks. The only exception is the first round of computation, since data can be specifically partitioned in such a way as to have all relevant data passed to the same map task. However, in a general setting, if in BSP two processors each send data to a third processor, this cannot be simulated using map tasks. Such direct communication can be achieved in the reduce phase since the shuffle step merges all the data with the same key and sends it to a single reduce task. The reduce task's id can be used as a key to achieve direct communication. Therefore, unless the algorithm only has two supersteps, the map phase is not used for computation.

\paragraph{}
Using only the reduce phase does not increase the theoretical costs of an algorithm, and given the simplicity provided by the MapReduce framework, incurring some extra cost is allowed in practice.  However, when certain BSP algorithms are implemented on the MapReduce framework, the asymptotic costs of the algorithm do increase. When a map or reduce task is finished, the local primary memory is cleared and the new map or reduce task's input data is loaded. This allows the framework to achieve automatic load balancing and fault tolerance. However, any data that needs to be used in later rounds has to be stored in global memory, inherently increasing the communication cost of the algorithm. If this extra cost is less than the algorithm's actual communication cost, then the algorithm can be efficiently implemented in MapReduce, otherwise it should not. Also, for certain BSP algorithms the communication cost is dominated by the cost of reading the input, with the cost in the following supersteps being significantly less. BSP algorithms do not cater for this feature, and the size of data stored on local memory for use in the following supersteps is not measured in the cost model.

\paragraph{}
A simple extension to the BSP model is proposed to allow classification of BSP algorithms into those that can and cannot be efficiently implemented on the MapReduce framework. $BSPMR(p,g,l)$ is a cost model exactly the same as BSP, with the same definitions for measuring an algorithm's efficiency, except for an extra cost $f_s$. This cost $f_s$ is the maximum size of data, over all the processors, that has to be stored in local memory, in superstep $s$, for use in later supersteps. The total cost of storing local data over all supersteps is $F = \sum_{s=1}^{S} f_s$. The cost of the algorithm is still $W + H \cdot g + S \cdot l$ when implemented on a conventional parallel system, but changes to $W + (H + F) \cdot g + S \cdot l$ when implemented on MapReduce.

\paragraph{}
Let the communication cost of a BSPMR algorithm not including the cost of reading the input and writing the output be $H_{n}$.

\begin{theorem}
\label{lemmaEfficiency}
Any $\mathrm{BSPMR}$ algorithm implementable on a $\mathrm{BSP}$ computer having $p$ processors with costs $W$, $H$, $F$, $S$ can be efficiently implemented in $\mathrm{MR}(v,g,l)$ with $v = p$ using the simulation provided in \cite{Goodrich} only if $(i)$ $O(T) = W \cdot p$, $(ii)$ $O(C) = H \cdot p$, $(iii)$ $O(D) = S$, and $(iv)$ $O(F) = H_{n}$.
\end{theorem}

\begin{proof}
$(i)$, $(ii)$, $(iii)$ Trivial.

$(iv)$ If the cost of reading the input on a BSP computer is $n$, the cost in MR is $O(n)$, since at most each value is given a key of constant size. The same applies to the output of the algorithm. This cost is included in $H$, but not in $H_{n}$, which only includes the cost of communicating local data between supersteps. Therefore, if $O(F) > H_{n}$, then the cost of implementing the algorithm on MapReduce is higher than implementing the algorithm on a conventional parallel system. \Square
\end{proof}

\paragraph{}
To better understand the use of MapReduce as an implementation framework for BSP algorithms a few examples are given below. These demonstrate the use of Theorem \ref{lemmaEfficiency} to determine whether MapReduce should be used to implement these algorithms.

\subsection{Sorting}
\label{sec:Sorting}
The problem of sorting a collection of comparable elements, generally taken to be numbers, has been extensively researched in the parallel algorithms domain. In BSP, two main algorithms exist. The first one, by Shi and Schaeffer \cite{ShiSchaeffer}, is known as Parallel Sorting by Regular Sampling (PSRS) and is asymptotically optimal for $n \geq p^3$. For $n < p^3$, Goodrich \cite{GoodrichSorting} proposed a different algorithm, the MapReduce version of which is given in \cite{Goodrich}.

\paragraph{}
The algorithm in \cite{ShiSchaeffer} is more suited for today's data sizes, and is very straightforward. It can be split up into two parts, sampling and sorting, and proceeds as follows. Given a collection of $n$ items, this is partitioned across the $p$ processors, each receiving about $\sfrac{n}{p}$ elements. Every processor sorts its data using an optimal sequential algorithm, and selects $p + 1$ regularly spaced primary samples, including the first and last elements, from this sorted data. A single processor then receives all $p \cdot (p + 1)$ primary samples, sorts them and again chooses $p + 1$ regular secondary samples from them. These secondary samples divide the $n$ input elements into $p$ buckets. Each processor is then assigned a bucket and every processor distributes its input data between the processors, based on their assigned bucket. Upon receiving all its data, each processor sequentially sorts it and outputs it. The algorithm has cost $W = O(\sfrac{(n \log n)}{p})$, $H = O(\sfrac{n}{p})$, $S = O(1)$. The extra cost measured in BSPMR is $F = O(\sfrac{n}{p})$, and $H_{n} = O(\sfrac{n}{p})$ so the algorithm can be efficiently implemented in MapReduce. 

\paragraph{}
In MapReduce the algorithm is just as straightforward and is again split into two parts. The algorithm can be implemented using the simulation described above, with each reduce task simulating a single processor in the BSP system. However, advantage can be taken of the MapReduce framework to allow for better implementation. Since input to MapReduce is in the form of $\langle key;value \rangle$ pairs, the input values are divided into $q$ partitions each having about $\sfrac{n}{q}$ elements, with each partition having a unique key from $1$ to $q$. The number of map tasks $q$ and reduce tasks $r$ should be set to the largest possible values allowed by the system, so we assume that $q \geq r$ as stated in \cite{MapReduceOrig} and \cite{HadoopTutorial}.

\paragraph{}
Every one of the $q$ available map tasks reads its input data, sorts the values and selects $r + 1$ regularly spaced primary samples. All the primary samples are sent to a single reduce task by setting the key to a common value, and are sorted by the shuffle function. The $r + 1$ regularly spaced secondary samples are then chosen, which determine the buckets to be used in the sorting phase. In the sorting phase, the input is again split amongst the available map tasks which produce for each input element $x_i$ a $\langle key;value \rangle$ pair with $key = value = x_i$. The shuffle function is changed so that every reduce task is assigned a bucket and every value is sent to the reduce task responsible for the bucket it falls into. All the values in each bucket are combined and automatically sorted by key, by the shuffle function, and the reduce function simply outputs the values in this sorted order. The cost of the algorithm in MR is $T = O(n \log n)$, $C = O(n)$, and $D = O(1)$.

\paragraph{}
Use of the MapReduce framework for sorting was first discussed in the original paper \cite{MapReduceOrig}. Hadoop has been used in sorting benchmarking\footnote{\url{http://sortbenchmark.org/}} \cite{TeraSort2008,TeraSort2009,TritonSort} with very good results, when sorting 1 Terabyte and 1 Petabyte of data. Google have also independently achieved very good results sorting large data \cite{GoogleSort}. The algorithms used for such sorting benchmarks are very similar to the one described above, varying mostly in the sampling phase, since deterministic sampling is replaced by random sampling.  This reduces the actual computation time, since no sorting is performed in the map phase of the sampling phase.

\subsection{Dense Matrix Multiplication}
Given $n \times n$ matrices $A$, $B$, the aim is to compute the $n \times n$ matrix $C = A \cdot B$. The problem is defined as
\[\begin{array}{ccc}
c_{ik} = \sum_{j=1}^{n} a_{ij} \cdot b_{jk}&  & 1 \leq i,k \leq n    
\end{array}\]
The standard sequential algorithm requires $\Theta(n^3)$ elementary operations.

\paragraph{}
A BSP algorithm for standard matrix multiplication is presented in \cite{McColl}, due to McColl and Valiant. The algorithm revolves around partitioning the input and output matrices across the $p$ available processors. The $n^3$ elementary products can be represented by a cube $V$ on axes $i, j, k$. Matrix $C$ can then be calculated by adding the elementary products in groups along the $j$-axis. The same operations can be performed in blocks, facilitating the distribution of all the elementary operations over the available processors.

\paragraph{}
Matrices $A,$ $B$ and $C$ are each divided into $p^{\sfrac{2}{3}}$ regular square blocks of size $\sfrac{n}{p\sfrac{1}{3}}$, with the blocks denoted by $A[i,j], B[j,k]$ and $C[i,k]$ for $1 \leq i, j, k \leq p^{\sfrac{1}{3}}$. Every processor is responsible for computing $V[i,j,k] = A[i,j] \cdot B[j,k]$, for a particular $i, j, k$, after receiving blocks $A[i,j]$ and $B[j,k]$ as input. The matrix $C$ is then divided between the $p$ processors and each of the $\sfrac{n^2}{p\sfrac{2}{3}}$ elements of the resulting $V[i,j,k]$ blocks is sent to the processor responsible for that summing up these elements to obtain $C[i,k]$. The cost of the algorithm is asymptotically optimal, with $W = O(\sfrac{n^3}{p})$, $H = O(\sfrac{n^2}{p^{\sfrac{2}{3}}})$ and $S = O(1)$. In BSPMR, $F = O(\sfrac{n^2}{p})$, with $H_{n} = O(\sfrac{n^2}{p^{\sfrac{2}{3}}})$ meaning the algorithm can be implemented efficiently on MapReduce.

\paragraph{}
Using the simulation presented above, the algorithm can be implemented in two rounds using $r = p$ reduce tasks to simulate the $p$ BSP processors. However, since the algorithm only runs in two supersteps, it can be implemented in MapReduce in a single round, using the whole framework, with $q = r = p$. Matrices $A, B, C$ are split into $q^{\sfrac{2}{3}}$ regular square blocks and with $A, B$ distributed amongst the $q$ available map tasks. Each map task is assigned the respective submatrices of $A$ and $B$, along with the indices of these submatrices. These indices are used as \textit{keys} to determine which of the $r$ reduce tasks will be sent the computed matrix operations. Each reduce task, upon receiving the elements, sums them up and outputs the respective submatrix of $C$. The number of map and reduce tasks are equal, with the algorithm's cost in MR being $W = O(n^3)$, $C = O(n^2 \cdot q^{\sfrac{1}{3}})$, and $D = O(1)$.

\paragraph{}
Matrix multiplication has been studied in the MapReduce context. As part of the Apache Software Foundation's efforts on Hadoop, HAMA was developed as a framework for massive matrix and graph computations using MapReduce and BSP, as detailed in \cite{Seo}.

\subsection{Breadth First Search}
The Breadth First Search (BFS) algorithm provides a simple way of searching for vertices in a graph $G = (V, E)$. The result is a tree, rooted at a vertex $s$, to all the vertices reachable from $s$. Each vertex's distance from $s$ is also noted. The typical approach taken for the parallelisation of the BFS algorithm \cite{Xia,Scarpazza,Yoo} is a level synchronous one, in line with the BSP model. The vertices of the graph are partitioned across the processors and then acted upon, with each processor also assigned the edges incident to each assigned vertex. A processor is only allowed to process the vertices it owns. Each vertex is marked as either processed or not processed. Initially, all the vertices are marked as not processed.

\paragraph{}
Given $G = (V, E)$, where $|V| \gg p$, the parallel BFS algorithm proceeds as follows. In the first superstep, the processor responsible for the root vertex identifies its neighbouring vertices and marks the root vertex as processed. Typically, most of the identified neighbouring vertices are not owned by the processor. Therefore, the processors responsible for each neighbouring vertex are identified and notified that these vertices are currently at the frontier of the search and need to be processed. In the following superstep, the next level in the BFS tree is processed. Each vertex reads the list of vertices that need to be processed in the current superstep. If any of these vertices is marked as not processed, their status is set to processed and their neighbours are identified. The processor responsible for each vertex is identified and notified and a new superstep can start. This process is repeated until all vertices have been processed. The cost of the algorithm is $W = O(\sfrac{|V|^2}{p})$, $H = O(\sfrac{|V|^2}{p})$, $S = O(d)$, where $d$ is the diameter of the graph $G$. The cost of not having storage, as measured in BSPMR, is $F = O(\sfrac{|V|^2}{p})$, while $H_{n} = O(|V|)$.

\paragraph{}
Since $F > H_{n}$, the algorithm cannot be efficiently implemented on MapReduce. When the primary local memory is cleared, the whole graph structure has to be stored and read from global memory in each superstep. Real life graphs exhibit the power law property \cite{Chakrabarti}. Therefore, the number of edges stored per processor varies depending on the assigned vertices. This makes load balancing very hard to achieve. For a graph with a fixed maximum degree $k$, $F = O(\sfrac{k|V|}{p})$, which could lead to an efficient MapReduce implementation for small values of $k$.

\paragraph{}
The inefficiency of implementing some graph algorithms on MapReduce is recognised in \cite{Seo}, which presents a framework that works in conjunction with Hadoop to perform matrix and graph computations. The authors claim that MapReduce is not appropriate for graph traversal algorithms, stating that algorithms based on their BSP engine would be better suited for such tasks.

\section{Conclusion}
The MapReduce framework has generated great interest in the area of parallel algorithm design due to the simplicity with which parallel algorithms can be developed. This wide use of the framework demands that it is put on sound theoretical foundations, which is essential if the potential and limitations of the framework are to be fully understood.

\paragraph{}
A number of theoretical models for the framework have already been proposed, but each poses limitations on the framework to enforce parallel algorithm design. These are not necessary for efficient design, as shown by previously existing models such as BSP. This work shows that the BSP model can in fact be used to model MapReduce algorithms, with the same asymptotic cost for conventional systems with identical number of processors. It is, therefore, determined that MapReduce should be used purely as an implementation framework with BSP used to design the algorithms.

\paragraph{}
This work also discusses the differences between MapReduce and BSP, highlighting the fact that BSP algorithms which require data to be stored and accessed in multiple supersteps cannot be efficiently implemented using MapReduce if the number of supersteps is not constant. This is due to the nature of the framework which does not allow for storage between rounds but requires that this data is communicated between such rounds. However, since the MapReduce framework provides a simple abstraction on top of the intricate parallel development challenges, it may still be deemed more feasible in practice to use the framework despite its inefficiencies.

\bibliography{bibliography} 

\begin{thebibliography}{10}

\bibitem{AmazonElastic}
{Amazon Web Services LLC}.
\newblock {Amazon Elastic MapReduce}.
\newblock \url{http://aws.amazon.com/elasticmapreduce/}, March 2012.

\bibitem{MapReduceAlg1}
{Amund Tveit}.
\newblock {MapReduce \& Hadoop Algorithms in Academic Papers}.
\newblock
  \url{http://atbrox.com/2011/05/16/mapreduce-hadoop-algorithms-in-academic-papers-4th-update-may-2011/},
  May 2011.

\bibitem{Hadoop}
{A}pache~{S}oftware {F}oundation.
\newblock {Hadoop MapReduce}.
\newblock \url{http://hadoop.apache.org/mapreduce/}, March 2012.

\bibitem{HadoopUsers}
{A}pache~{S}oftware {F}oundation.
\newblock Hadoop {W}iki: {P}owered{B}y.
\newblock \url{http://wiki.apache.org/hadoop/PoweredBy}, March 2012.

\bibitem{HadoopTutorial}
{A}pache~{S}oftware {F}oundation.
\newblock {MapReduce Tutorial}.
\newblock
  \url{http://hadoop.apache.org/common/docs/current/mapred_tutorial.html},
  March 2012.

\bibitem{Chakrabarti}
Deepayan Chakrabarti, Yiping Zhan, and Christos Faloutsos.
\newblock {R-MAT: A Recursive Model for Graph Mining}.
\newblock In {\em {SIAM International Conference on Data Mining}}, 2004.

\bibitem{GoogleSort}
Grzegorz Czajkowski.
\newblock Sorting 1pb with mapreduce.
\newblock
  \url{http://googleblog.blogspot.com/2008/11/sorting-1pb-with-mapreduce.html},
  November 2008.

\bibitem{MapReduceOrig}
Jeffrey Dean and Sanjay Ghemawat.
\newblock {MapReduce: Simplified Data Processing on Large Clusters}.
\newblock In {\em Proceedings of the 6th Conference on Symposium on Operating
  Systems Design \& Implementation - Volume 6}, pages 137--150. ACM, 2004.

\bibitem{Feldman}
Jon Feldman, S.~Muthukrishnan, Anastasios Sidiropoulos, Cliff Stein, and Zoya
  Svitkina.
\newblock {On Distributing Symmetric Streaming Computations}.
\newblock In {\em Proceedings of the 19th Annual ACM-SIAM Symposium on Discrete
  Algorithms}, SODA '08, pages 710--719. ACM, 2008.

\bibitem{Goodrich}
M.~T. {Goodrich}, N.~{Sitchinava}, and Q.~{Zhang}.
\newblock {Sorting, Searching, and Simulation in the MapReduce Framework}.
\newblock {\em ArXiv e-prints, \url{http://arxiv.org/abs/1101.1902v1}}, January
  2011.

\bibitem{GoodrichSorting}
Michael~T. Goodrich.
\newblock Communication-efficient parallel sorting.
\newblock {\em SIAM Journal on Computing}, 29:416--432, October 1999.

\bibitem{Graham}
R.~L. Graham.
\newblock {Bounds for certain multiprocessing anomalies}.
\newblock {\em Bell System Technical Journal}, 45:1563--1581, November 1966.

\bibitem{Hochbaum}
Dorit~S. Hochbaum and David~B. Shmoys.
\newblock Using dual approximation algorithms for scheduling problems
  theoretical and practical results.
\newblock {\em J. ACM}, 34:144--162, January 1987.

\bibitem{JaJa}
Joseph J{\'{a}}J{\'{a}}.
\newblock {\em {An Introduction to Parallel Algorithms}}.
\newblock Addison Wesley Longman Publishing Co., Inc., 1992.

\bibitem{Karloff}
Howard Karloff, Siddharth Suri, and Sergei Vassilvitskii.
\newblock {A Model of Computation for MapReduce}.
\newblock In {\em Proceedings of the 21st Annual ACM-SIAM Symposium on Discrete
  Algorithms}, SODA '10, pages 938--948. ACM, 2010.

\bibitem{Lisp}
{LispWorks Ltd.}
\newblock {Common Lisp Documentation}.
\newblock \url{http://www.lispworks.com/documentation/common-lisp.html}, March
  2012.

\bibitem{McColl}
W.~McColl.
\newblock Scalable {C}omputing.
\newblock In Jan van Leeuwen, editor, {\em Computer Science Today}, volume 1000
  of {\em Lecture Notes in Computer Science}, pages 46--61. Springer Berlin /
  Heidelberg, 1995.

\bibitem{MapReduceAlg2}
{Mendeley Ltd.}
\newblock {Papers in MapReduce Applications}.
\newblock
  \url{http://www.mendeley.com/groups/1058401/mapreduce-applications/papers/},
  March 2012.

\bibitem{Daytona}
{Microsoft Corporation}.
\newblock {Daytona}.
\newblock
  \url{http://research.microsoft.com/en-us/projects/daytona/default.aspx},
  March 2012.

\bibitem{TeraSort2008}
Owen O'Malley.
\newblock Terabyte sort on hadoop.
\newblock \url{http://sortbenchmark.org/YahooHadoop.pdf}, May 2008.

\bibitem{TeraSort2009}
Owen O'Malley and Arun~C. Murthy.
\newblock Winning a 60 second dash with a yellow elephant.
\newblock \url{http://sortbenchmark.org/Yahoo2009.pdf}, April 2009.

\bibitem{TritonSort}
Alexander Rasmussen, Michael Conley, George Porter, and Amin Vahdat.
\newblock Tritonsort 2011.
\newblock \url{http://sortbenchmark.org/2011_06_tritonsort.pdf}, June 2011.

\bibitem{Scarpazza}
D.P. Scarpazza, O.~Villa, and F.~Petrini.
\newblock {Efficient Breadth-First Search on the Cell/BE Processor}.
\newblock {\em IEEE Transactions on Parallel and Distributed Systems},
  19(10):1381--1395, October 2008.

\bibitem{Seo}
Sangwon Seo, Edward~J. Yoon, Jaehong Kim, Seongwook Jin, Jin-Soo Kim, and
  Seungryoul Maeng.
\newblock {HAMA: An Efficient Matrix Computation with the MapReduce Framework}.
\newblock {\em IEEE International Conference on Cloud Computing Technology and
  Science}, pages 721--726, 2010.

\bibitem{ShiSchaeffer}
Hanmao Shi and Jonathan Schaeffer.
\newblock Parallel sorting by regular sampling.
\newblock {\em Journal of Parallel and Distributed Computing}, 14:361--372,
  April 1992.

\bibitem{Sipser}
Michael Sipser.
\newblock {\em Introduction to the Theory of Computation; 2nd ed.}
\newblock Thomson Course Technology, 2006.

\bibitem{Valiant}
Leslie~G. Valiant.
\newblock {A Bridging Model for Parallel Computation}.
\newblock {\em Communications of the ACM}, 33:103--111, August 1990.

\bibitem{Xia}
Yinglong Xia and Viktor~K. Prasanna.
\newblock {Topologically Adaptive Parallel Breadth-First Search on Multicore
  Processors}.
\newblock In {\em Proceedings of the 21st International Conference on Parallel
  and Distributed Computing and Systems}, PDCS'09, Cambridge, MA, November
  2009.

\bibitem{Yoo}
Andy Yoo, Edmond Chow, Keith Henderson, William McLendon, Bruce Hendrickson,
  and Umit Catalyurek.
\newblock {A Scalable Distributed Parallel Breadth-First Search Algorithm on
  BlueGene/L}.
\newblock {\em Super Computing Conference}, 2005.

\end{thebibliography}
\bibliographystyle{plain}

\end{document}